\documentclass[notitlepage,jmp,10pt,superscriptaddress,aps,nofootinbib]{revtex4-1}

\usepackage{amsmath}
\usepackage{amssymb}
\usepackage{graphicx}
\usepackage{hyperref}
\usepackage{color}
\usepackage{mathrsfs}
\usepackage{extarrows}

\usepackage[T1]{fontenc}
\usepackage[utf8]{inputenc}

\usepackage{amsthm}
\usepackage{bbm}
\usepackage{mathtools}

\newtheorem{thm}{Theorem}[section]

\newtheorem{lem}[thm]{Lemma}
\newtheorem{prop}[thm]{Proposition}
\newtheorem{cor}[thm]{Corollary}
\newtheorem{remark}[thm]{Remark}

\def\tr{\operatorname{tr}}
\def\idty{\mathbbm1} %  unit operator

\def\Cx{{\mathbb C}}
\def\Ir{{\mathbb Z}}\def\Nl{{\mathbb N}}
\def\Rt{{\mathbb Q}}

\def\braket#1#2{\langle #1,#2\rangle}

\def\brAAket#1#2#3{\langle#1\vert#2\vert#3\rangle}

\def\ket #1{\vert#1\rangle}

\def\tr{\mathop{\rm tr}\nolimits}
\def\abs#1{\vert#1\vert}

\def\BB{{\mathcal B}}\def\HH{{\mathcal H}}

\def\BH{\BB(\HH)}

\def\RA{\mathcal A_\MF}
\def\vNA{\mathscr{W}_\MF}

\def\embf{\emph} 

\def\MF{\Phi} %magnetic field
\def\nr{p} %numerator in cf expansion
\def\dr{q} %denominator in cf expansion
\def\Wm{W_{\MF}}  %magnetic walk
\def\wm{w_{\MF}}    %magnetic walk in algebra
\def\Wmt{\widetilde{W}_{\MF}}  %magnetic walk mit tilde
\def\Wz{W_{0}}  %electric walk with zero field, ie. non electric walk
\DeclareMathOperator\dist{dist}

\def\idty{{\leavevmode\rm 1\mkern -5.4mu I}} %  unit operator

\def\Cx{{\mathbb C}}
\def\Ir{{\mathbb Z}}\def\Nl{{\mathbb N}}
\def\Rt{{\mathbb Q}}\def\torus{{\mathbb T}}
\def\mm{\mathbb M}

\def\ket #1{\vert #1\rangle}
\def\braket #1#2{\langle #1 \vert #2\rangle}

\def\braketop #1#2#3{\langle #1 \vert #2\vert #3\rangle}
\def\abs#1{\vert#1\vert}

\def\inv{{-1}}

\usepackage{tikz}
\usetikzlibrary{shapes,arrows,positioning,calc}

\begin{document}

\title{Singular continuous Cantor spectrum for magnetic quantum walks}

\author{C. Cedzich}
\affiliation{Laboratoire de Recherche en Informatique (LRI), Universit\'{e} Paris Sud, CNRS, Centrale Sup\'{e}lec, \\B\^{a}t 650, Rue Noetzlin, 91190 Gif-sur-Yvette, France}
\author{J. Fillman}
\affiliation{Department of Mathematics, Texas State University, San Marcos, TX 78666, USA}
\author{T. Geib}
\affiliation{Institut f\"ur Theoretische Physik, Leibniz Universit\"at Hannover, Appelstr. 2, 30167 Hannover, Germany}
\author{A.~H. Werner}
\affiliation{{QMATH}, Department of Mathematical Sciences, University of Copenhagen, Universitetsparken 5, 2100 Copenhagen, Denmark,}
\affiliation{{NBIA}, Niels Bohr Institute, University of Copenhagen, Denmark}

\begin{abstract}
 In this note, we consider a physical system given by a two-dimensional quantum walk in an external magnetic field. In this setup, we show that both the topological structure as well as its type depend sensitively on the value of the magnetic flux $\Phi$: while for $\Phi/(2\pi)$ rational the spectrum is known to consist of bands, we show that for $\MF/(2\pi)$ irrational the spectrum is a zero-measure Cantor set and the spectral measures have no pure point part.
\end{abstract}

\maketitle

\section{Introduction}

Electrons moving in a two-dimensional lattice under the influence of a uniform magnetic field are among the most intensively studied physical systems of the last decades, see e.g. \cite{brown1964bloch,hof76,TKNN,Laughlin}. The Almost-Mathieu Operator (and more generally the Extended Harper's Model) is one example of such a model which inspired intense study in mathematics; see \cite{AJ2009Ann,AJM2017Inv,MJ2017ETDS,J1999Ann} and references therein. Despite the simplicity of the setup its dynamical properties are surprisingly rich, and the propagation behaviour as well as the spectrum depend crucially on properties of the magnetic field like the rationality of its flux.

A similarly drastic dependence of the spectral and dynamical properties on a system parameter has recently been observed in quantum walks subject to external electric fields \cite{ewalks,anders_loc}. Quantum walks are discrete-time analogues of the unitary dynamics of a single particle with internal degree of freedom on a lattice constrained to finite propagation speed per time step \cite{aharonov1993quantum,Ambainis2001,Grimmet,TRcoin}. Also, quantum walks have shown to be a useful model for quantum simulation encompassing a variety of single-particle and few-particle quantum effects such as ballistic transport, decoherence, Anderson localization and graphene-like dispersion relations as well as the formation of bound states and symmetry protected topological phases \cite{dynlocalain,Joye2011,Joye_Merkli,dynloc,F2017CMP,molecules,moleculestheothers,SpaceTimeCoinFlux,TRcoin,UsOnTop_short,UsOnTop_long,UsOnTI}. At the same time quantum walks and their time-continuous counterpart are relevant in a quantum algorithmical context, for example in search algorithms, to examine the distinctness of elements, in quantum information processing and applications to the graph isomorphism problem \cite{childs2003exponential, ambainis2003quantum,Lovett:2010ff, childs2009universal,berry2011two,portugal2018quantum}. From  a mathematical point of view, quantum walks on a one-dimensional lattice can also be seen as a particular class of CMV matrices giving a link between quantum dynamical systems and orthogonal polynomials on the unit circle, a connection which has proved to be fruitful in both directions \cite{BGVW,CGMV,CGMV2,GVWW,WeAreSchur,cedzich2015quantum}. One can use spectral methods to deduce bounds on spreading \cite{DFO,DFV}. On the other hand, estimates of a dynamical nature can be turned into estimates on spectral quantities, for example, quantitative regularity of the spectral measures \cite{DEFHV}.

Recently, the coupling of quantum walks to electromagnetic fields was implemented by extending the minimal coupling principle to the discrete setting \cite{UsOnELM}, which allows one to analyze the reaction of a given quantum walk dynamics to an electromagnetic field. As already mentioned, for one dimensional systems the dynamics depends crucially on the rationality of the electric field: For rational fields the spectrum is absolutely continuous and the propagation is eventually ballistic after an initial segment with revivals of any initial state similar to Bloch oscillations. If the electric field is irrational the characterization is more involved and depends on the approximability of the field in terms of continued fractions: for well approximable fields the spectrum is singular continuous and the propagation is hierarchical, i.e. there is an infinite sequence of better and better revivals alternating with farther and farther excursions \cite{ewalks}. On the other hand, for badly approximable fields the system displays Anderson localization which even turns out to be the typical case \cite{anders_loc}.

Going beyond one spatial dimension, for the Hamiltonian setting, it is possible to connect the spectrum of the two dimensional magnetic Laplacian with the spectrum of the almost Mathieu operator, allowing to analyze not only the spectral type, but also its topological structure.
Taking these questions as a starting point, we discuss in this paper in detail two-dimensional quantum walks which we place into homogeneous magnetic fields using the aforementioned discrete analogue of minimal coupling from \cite{UsOnELM}. Also in the magnetic case the spectral and dynamical properties depend crucially on the rationality of the field, but this dependence turns out to be slightly less involved than for the electric walks in \cite{ewalks,anders_loc}: for rational fields the system is effectively translation invariant which by general arguments implies absolutely continuous spectrum and ballistic transport. For irrational fields, the spectrum is purely singular continuous which by the RAGE theorem \cite{cyconschroeder} excludes any form of localization and thereby contrasts the typicality of Anderson localization for one-dimensional electric walks; see also \cite{FO2017JFA} for a proof of the RAGE theorem in the discrete-time (unitary) case.

The paper is organized as follows: in Section \ref{sec:thesystem} we define the system under consideration and state the main result in
Theorem \ref{thm:specWm}. The proof of this theorem is divided into Section \ref{sec:proof_cantor} and Section \ref{sec:proof_absence}. 
\section{The system} \label{sec:thesystem}
\subsection{The physical model}

Quantum walks describe the time-discrete evolution of a single particle with an internal degree of freedom on a
lattice under the additional assumption of a finite propagation speed. In this paper, we consider particles on the two-dimensional lattice $\Ir^2$ with two-dimensional internal degree of freedom $\Cx^2$. Basis vectors of the corresponding Hilbert space $\HH=\ell^2(\Ir^2)\otimes\Cx^2$ are denoted by $\ket{x,s}$ with $x\in\Ir^2$ and $s=\pm1$, and the coordinates of $\psi \in \HH$ with respect to this basis are denoted by $\psi(x,\pm) \equiv \braket{x,\pm1}\psi$. A single timestep is implemented by the unitary operator
\begin{equation}\label{eq:dirac_walk}
  \Wz=S_1C_1S_2C_2,
\end{equation}
where the \embf{coin operators} $C_j=\bigoplus_{x\in\Ir^2}C_j(x)$ locally rotate the internal degree of freedom as
\begin{equation*}
  C_j(x)=\begin{pmatrix} c^j_{11}(x)    &   c^j_{12}(x) \\  c^j_{21}(x) &   c^j_{22}(x) \end{pmatrix}.
\end{equation*}
If not otherwise noted, in this paper we assume the coins to be translation invariant which means they act the same everywhere. Concretely, we take $C_1=C_2=\idty\otimes C_H$ where $C_H$ is the two-dimensional Hadamard matrix
\begin{equation}\label{eq:walk}
  C_H=\frac1{\sqrt2}\begin{pmatrix}1&\hphantom-1\\1&-1\end{pmatrix}.
\end{equation}

In contrast to the coin operators, the \embf{state-dependent shifts} $S_\alpha,\:\alpha=1,2,$ relate neighbouring cells. They are defined by $S_\alpha=\sum_{s=\pm1}t_\alpha^s\otimes P_s$ where $t_\alpha\ket{x}=\ket{x+e_\alpha}$ denotes the lattice translation in direction $e_\alpha$ and $P_{\pm1}$ projects onto the internal basis state $\ket{\pm1}$.

The aim of this paper is to study the walk $\Wz$ in external magnetic fields. Since $\Wz$ is given as a sequence of coin and shift operators, such fields are introduced by an approach similar to the minimal coupling scheme in continuous time \cite{UsOnELM}. This amounts to replacing the lattice translations $t_\alpha$ in the definition of the $S_\alpha$ by the unitary \embf{magnetic translations}
\begin{equation}\label{eq:magshift}
  T_\alpha\ket{x}=t_\alpha U_{\alpha}(x)\ket{x}=U_{\alpha}(x)\ket{x+e_\alpha}.
\end{equation}
Here, the $U_\alpha(x)$ are position- and direction-dependent phases which are determined only up to a discrete gradient $G(x+e_\alpha) G(x)^\inv$ reflecting the freedom of local $U(1)$-gauge transformations $G:\ket x\mapsto G(x)\ket x$.  In contrast to the lattice translations $t_\alpha$, the magnetic translations $T_\alpha$ do not commute anymore. Instead, they commute up to a plaquette phase, i.e.
\begin{equation}\label{eq:field}
  T_\alpha^*T_\beta^*T_\alpha T_\beta\ket x=U_\alpha(x)^*U_\beta(x+e_\alpha)^*U_\alpha(x+e_\beta)U_\beta(x)\ket x=e^{-iF_{\alpha\beta}(x)}\ket x.
\end{equation}
This commutator is invariant under local gauge transformations and we naturally call $F_{\alpha\beta}=-F_{\beta\alpha}\in[0,2\pi[$ a \embf{discrete magnetic field} \cite{UsOnELM}. Throughout this manuscript, we assume this magnetic field to be \embf{homogeneous}, i.e. independent of $x$, and denote it by $F_{12}=:\MF$. Note that since there are no electric fields present, a discrete version of the homogeneous Maxwell equations implies that $\MF$ is not only homogeneous but also static, i.e. independent of time \cite{UsOnELM}. If not otherwise stated, in this manuscript we fix the gauge to the \embf{symmetric gauge}
\begin{equation} \label{eq:symmetricGauge}
  U_1(x)=e^{-i x_2\MF/2},\quad U_2(x)=e^{i x_1\MF/2}.
\end{equation}

After applying this discrete minimal coupling procedure to $\Wz$ a single timestep is implemented by the \embf{magnetic walk}
\begin{equation}\label{eq:dirac_mag_walk}
  \Wm=\begin{pmatrix} T_1 & \\ & T_1^* \end{pmatrix} C_H \begin{pmatrix} T_2 & \\ & T_2^* \end{pmatrix} C_H.
\end{equation}

\begin{remark}\label{rem:sym_gauge_coins}
  The choice of gauge in \eqref{eq:symmetricGauge} ensures that the phase $U_\alpha(x)$ commutes with $t_\alpha$ in \eqref{eq:magshift} which allows us to interpret the magnetic walk $\Wm$ as a walk of the form \eqref{eq:dirac_walk} with quasi-periodic coins
\begin{equation} \label{eq:symGaugeQPcoins}
C_j
= \bigoplus_{x\in\Ir^2} e^{(-1)^j i\MF x_k \sigma_3 /2}\cdot C_H,\qquad j,k =1,2,\ j\neq k,
  \end{equation}
  where $\sigma_3$ is the third Pauli matrix.
Let us note for later use that the coins in \eqref{eq:symGaugeQPcoins} are never purely (off-)diagonal.
\end{remark}
For rational fields $\MF/(2\pi)=\nr/\dr$ the magnetic walk $\Wm$ is translation invariant after regrouping and can thus be diagonalized in momentum space \cite{UsOnRat_fields}. Its spectrum is absolutely continuous and consists of $2\dr$ bands which resemble a two-fold copy of the Hofstadter butterfly, see Figure \ref{Butterfliege}. For $\MF/(2\pi)$ irrational, arguments similar to those in Hofstadter's original paper \cite{hof76} immediately lead one to conjecture Cantor spectrum for $\Wm$. It is the goal of this paper to prove this conjecture and characterize the spectral type of $\Wm$ for $\MF/(2\pi)$ irrational.

\begin{figure}[t]
\begin{center}
  \includegraphics[width=.4\textwidth]{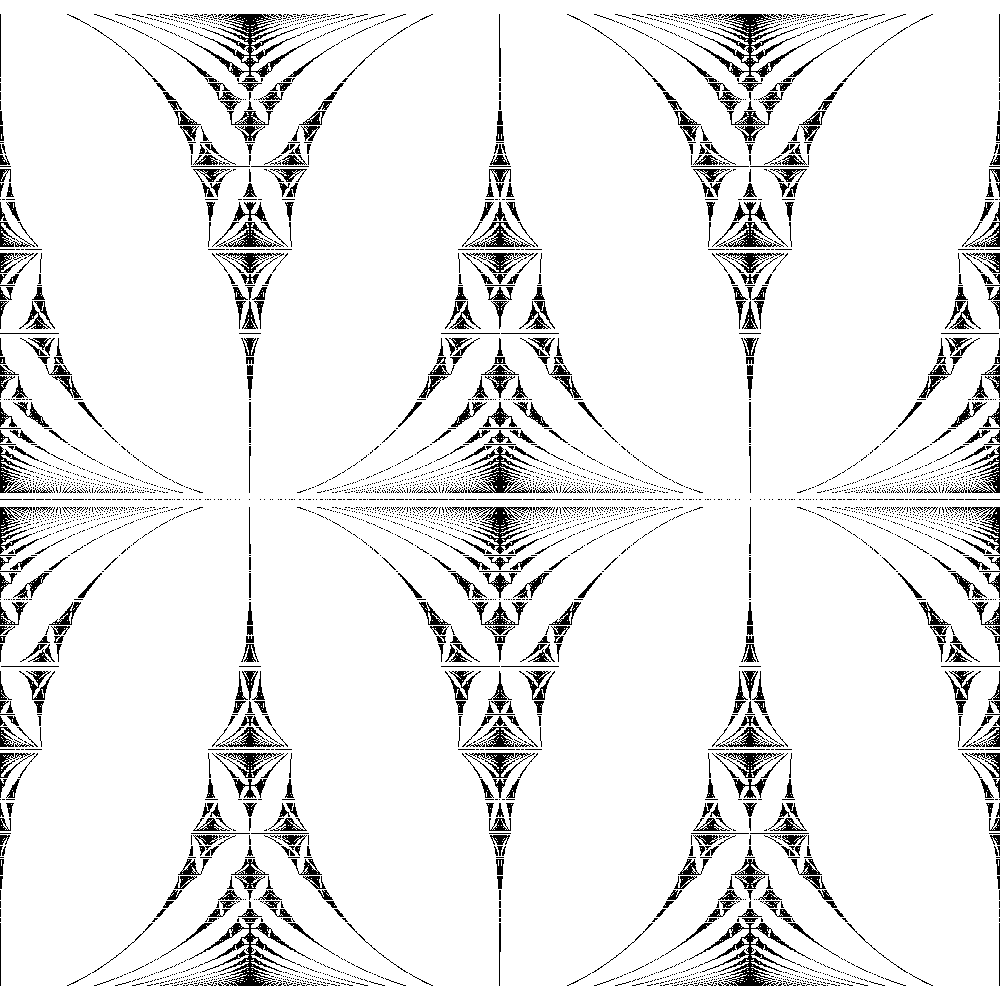}
  \caption{\label{Butterfliege}The spectra of the two-dimensional magnetic walk \eqref{eq:dirac_mag_walk} for rational fields  \cite{christalk} resembles a two-fold copy of the Hofstadter butterfly. The vertical axis corresponds to the field $\MF$ (plotted up to denominator 70) and the horizontal axis to the argument of the quasi-energy. Its fractal structure is evident. The symmetry of the spectrum with respect to $\MF\mapsto-\MF$ follows from general arguments about the rotation algebra \cite{Rieffel_Irrat_Rot}, and the symmetry with respect to $z\mapsto z^*$ is a consequence of $K\Wm K=W_{-\MF}$, where $K$ denotes complex conjugation, together with $\sigma(\Wm)=\sigma(W_{-\MF})$. Finally, the symmetry with respect to $z \mapsto -z$ is a consequence of $JW_\MF J = - W_\MF$, where $J$ is the unitary involution  $J\ket{x,s} = (-1)^{x_1+x_2}\ket{x,s}$.
  }
\end{center}
\end{figure}

\subsection{The rotation algebra point of view}
In order to get a handle on the spectral properties of $\Wm$ we sometimes adopt an algebraic approach. The magnetic translations $T_1$ and $T_2$ realize a particular instance of the so-called \emph{rotation algebra} $\RA$ which is defined abstractly as the universal $C^*$-algebra generated by two unitary elements $u,v$ such that
\begin{equation}\label{eq:uv}
  uv=e^{-i\MF}vu.
\end{equation}
This correspondence is established by the representation $\pi_{2}:\RA\to\BB(\ell^2(\Ir^2))$ given by the magnetic translations
\begin{equation}\label{eq:magrep}
  \pi_{2}(u)=T_1,\quad\pi_{2}(v)=T_2.
\end{equation}
Let us point out that the subscript on the representation refers to the dimension of the lattice in the codomain.

Denoting by $\mm_2(\RA)= \mm_2(\Cx) \otimes \RA$ the ring of $\RA$-valued $2\times2$-matrices, this representation of $\RA$ can be lifted to a map $\Pi_2:\mm_2(\RA)\to\mm_2(\BH)$ which acts on $a=\sum_{i,j=1,2}e_{ij}\otimes a_{ij}\in\mm_2(\RA)$ as
\begin{equation}\label{eq:Pi2}
  \Pi_2(a)=\sum_{i,j=1,2}e_{ij}\otimes \pi_2(a_{ij}),
\end{equation}
where $e_{ij}$ denotes the matrix whose $(i,j)$ entry is 1 and all others are $0$. Clearly, $\Pi_2$ is a $*$-representation because $\pi_2$ is. Then, we can interpret the magnetic walk $\Wm$ as the image under $\Pi_2$ of the unitary element $\wm\in\mm_2(\RA)$ given by
\begin{equation}\label{eq:abstract_walk}
  \wm=\begin{pmatrix}
    u& \\ & u^*
  \end{pmatrix}
  C_H\begin{pmatrix}
    v& \\ & v^*
  \end{pmatrix}C_H
  =\frac12\begin{pmatrix}
    u(v+v^*)    &   u(v-v^*)    \\
    u^*(v-v^*)    &   u^*(v+v^*)
  \end{pmatrix}.
\end{equation}

The abstract point of view in terms of the rotation algebra proves helpful in establishing the main result of this paper. Indeed, the first statement in Theorem \ref{thm:specWm} below is a purely topological result which is independent of the particular representation of the rotation algebra.

\begin{remark}
  The algebraic point of view on homogeneous translation systems extends to magnetic walks on higher dimensional lattices. However, it is not helpful in the presence of electric fields: following the construction in \cite{UsOnELM}, electric fields imply that the spatial translations $T_\alpha$ in \eqref{eq:magshift} commute merely up to phases with the time translation $T_0$ in the equation of motion. However, since $T_0$ relates different ``time slices'' of $\HH$ it cannot be represented as a Hilbert space operator.
\end{remark}

\subsection{The main result}

In this note we prove that the spectrum of $\Wm$ is a Cantor set for $\MF/(2\pi)$ irrational and, moreover, show that it cannot have a pure point part. To formulate this precisely:
\begin{thm}\label{thm:specWm}
  Let $\Wm$ be the magnetic walk defined in \eqref{eq:dirac_mag_walk}. Then:
  \begin{enumerate}
    \item   For $\MF/(2\pi)\notin\Rt$ the spectrum of $\Wm$ is a zero-measure Cantor set.
    \item   For any magnetic field the pure point spectrum of $\Wm$ is empty.
  \end{enumerate}
\end{thm}
\noindent The zero-measure statement in the theorem refers to the arc-length measure on the circle. When we say that a subset $K$ of the circle is a Cantor set, we mean that $K$ is a closed, perfect, nowhere dense subset of the circle in the standard topology thereupon. Together, these results imply the following:
\begin{cor}
  For $\MF/(2\pi)$ irrational, the spectrum of $\Wm$ is purely singular continuous.
\end{cor}
This follows immediately: on the one hand, the spectrum has zero Lebesgue measure and therefore cannot have an absolutely continuous component, and on the other also the pure point component is empty by the second statement in the theorem above.

The spectral characterization of $\Wm$ in this theorem is similar to the one obtained in \cite{FOZ17} for a related model. Indeed, the proof of Cantor spectrum builds on the corresponding result in \cite{FOZ17} and was inspired by \cite{shubinDML}. Yet, the proof of singular continuous spectrum does not carry over straightforwardly: the model in \cite{FOZ17} is a quasi-periodic walk on the integers and the proof of the absence of pure point spectrum makes use of transfer matrix techniques. In two spatial dimensions such transfer matrix techniques are no longer available, and we have to find different techniques to exclude eigenvalues: taking an abstract algebraic point of view is inspired by ideas from the self-adjoint setting \cite{shubinDML}.

The methods we use below in the proof of Cantor spectrum are restricted to the magnetic walk with local coin given by the Hadamard matrix $C_H$. It is an interesting question whether the statement holds also if we allow for more general coins. While we suspect it does, a proof seems to require more heavy machinery and certainly goes beyond the scope of the current project. In contrast, the proof of the absence of point spectrum below is independent of the concrete choice of the coins: it applies to magnetic walks of the form \eqref{eq:dirac_walk} with arbitrary coins except for purely diagonal and purely off-diagonal ones. 
\section{Proof of Cantor spectrum}\label{sec:proof_cantor}

In this section we prove Cantor spectrum for the two-dimensional magnetic walk $\Wm$ with $\MF/(2\pi)$ irrational, i.e. the first statement of Theorem \ref{thm:specWm}. The idea behind the proof is to relate $\Wm$ to the so-called \emph{unitary critical almost Mathieu operator} which is defined as the quantum walk $\Wmt$ on $\ell^2(\Ir)\otimes\Cx^2$ given by
\begin{equation}
	\Wmt=SC_2(\MF),
\end{equation}
where $S$ denotes the conditional shift $S\ket{x,s}=\ket{x+s,s}$ for $x\in\Ir$, $s=\pm1$ and $C_2(\MF)$ is a position-dependent coin acting locally as
\begin{equation}\label{eq:shikanowalk}
	C_2(\MF,x)=e^{i(\MF x+\theta)\sigma_2}=\begin{pmatrix}
	\cos(\MF x+\theta) & -\sin(\MF x+\theta)\\
	\sin(\MF x+\theta) & \hphantom-\cos(\MF x+\theta)
	\end{pmatrix}
\end{equation}
for some fixed $\theta\in\torus$. Here, the subscript in the definition of $C_2(\MF)$ indicates the rotation axis in $\Cx^2$.
This walk was introduced in \cite{Linden2009,Shikano:2010id,Shikano:2011br} as a model for an inhomogeneous quantum walk, and its spectral properties have been studied extensively in \cite{FOZ17}. In particular \cite[Thm.\,1.1\,c]{FOZ17}:
\begin{thm}
  Let $\MF/(2\pi)$ be irrational. Then, the spectrum of $\Wmt$ is a zero measure Cantor set for all $\theta$.
\end{thm}

Crucial for the proof of this theorem are the transfer matrix techniques for one-dimensional systems, which are not available on the two-dimensional lattice. However, we can establish a relation between the one-dimensional walk $\Wmt$ and the two-dimensional walk $\Wm$ by considering the representation $\pi_1$ of $\RA$ on $\ell^2(\Ir)$ given by
\begin{equation}\label{eq:choirep}
  \pi_1(u)\ket x=\ket{x+1},\qquad \pi_1(v)\ket x=e^{i (x\MF+\theta)}\ket x,
\end{equation}
where as above the subscript on the representation refers to the dimension of the lattice in the codomain. Similar to the definition of $\Pi_2$ in \eqref{eq:Pi2} this representation can be lifted to a representation $\Pi_1$ of $\mm_2(\RA)$. The image of $\wm$ in this lifted representation is given by
\begin{equation}
  \Pi_1(\wm)=SC_1(\MF),
\end{equation}
where $C_1(\MF)$ is a coin operator acting locally as $C_1(\MF,x)=\exp(i\MF x\sigma_1)$. In particular, $C_1(\MF)$ is unitarily equivalent to $C_2(\MF)$ by $\idty\otimes\exp(i\sigma_3\pi/4)$ which commutes with the shift operator $S$. Thus, up to unitary equivalence, $\Wm$ and $\Wmt$ are the images of the same element $\wm$ of $\mm_2(\RA)$ under different representations.

The first statement in Theorem \ref{thm:specWm} then follows from general observations about $\RA$ and $\mm_2(\RA)$, or more general $\mm_n(\RA)$, and their images under different representations. First, for $\MF/(2\pi)$ irrational $\RA$ is simple, i.e. its only closed, two-sided ideals are trivial \cite{Rieffel_Irrat_Rot}. This implies that $\mm_n(\RA)$ is also simple for all $n$ \cite{grillet2007abstract}.
The simplicity of a $C^*$-algebra $\mathcal A$ has profound consequences: it implies that all its representations are faithful, hence isometric and therefore preserve the spectrum \cite{BratelliRobinson}:
\begin{prop}
  Let  $\pi,\pi'$ be representations of a simple $C^*$-algebra $\mathcal A$. Then, for any $a\in\mathcal A$
  \begin{equation}
    \sigma(\pi(a))=\sigma(\pi'(a)).
  \end{equation}
\end{prop}
\noindent Thus, we conclude:
\begin{cor}
  Let $\MF/(2\pi)$ be irrational. Then $\Wm$ and $\Wmt$ have the same spectrum as a subset of $\Cx$, i.e.
  \begin{equation}
	\sigma\bigl(\Wm\bigr)=\sigma\bigl(\Pi_2(\wm)\bigr)=\sigma\bigl(\Pi_1(\wm)\bigr)=\sigma\bigl(\Wmt\bigr).
  \end{equation}
\end{cor}
\section{Absence of point spectrum}\label{sec:proof_absence}

The present section concludes the proof of Theorem~\ref{thm:specWm} by proving the second statement. Let us note that if $\MF/(2\pi)$ is rational, then the spectral type of $W_\MF$ is purely absolutely continuous and hence point spectrum is absent. Consequently, we focus on the case $\MF/(2\pi)\notin \Rt$. We will again use some ideas from the rotation algebra perspective. Choose the symmetric gauge $(U_1,U_2)$ as in Section~\ref{sec:thesystem} and let $\vNA = \vNA(T_1,T_2) \equiv W^*(T_1,T_2) \subseteq \BB(\ell^2(\Ir^2))$ denote the von Neumann algebra generated by the magnetic translations $T_1$ and $T_2$.

Let us denote the center of a von Neumann-algebra $\mathscr{A}$ by $Z(\mathscr{A})$, that is, $Z(\mathscr{A})=\{a\in \mathscr{A}:ab=ba\:\forall b\in \mathscr{A}\}$. Additionally, for $n \in \Nl$, we define $\mm_n(\mathscr{A})=\mm_n(\Cx)\otimes\mathscr{A}$ to be the matrix algebra of $n\times n$ matrices with entries in $\mathscr{A}$. We first note that the center of the matrix algebra over $\vNA$ is trivial:
\begin{lem}
  If $\MF/(2\pi)$ is irrational, the center of $\mm_n(\vNA)$ is trivial, that is
  \begin{equation}
  Z(\mm_n(\vNA)) =
  \big\{\idty_n \otimes a : a \in Z(\vNA)=\{\lambda \idty:\lambda\in\Cx\}\big\},
  \end{equation}
  where $\idty_n$ and $\idty$ denote the identity in $\mm_n(\Cx)$ and in $\vNA$, respectively.
\end{lem}
\begin{proof}
It is straightforward to show that the center of a matrix ring over a unital algebra $\mathscr{A}$ consists of scalar multiples of the identity matrix with scalars belonging to the center of $\mathscr{A}$. A standard reference is \cite[Chapter IX, Corollary 4.4]{grillet2007abstract}. The statement thus follows from the fact that the center of $\vNA$ is trivial if $\MF/(2\pi)$ is irrational \cite[Proposition~2.2]{shubinDML}.
\end{proof}

Naturally, the foregoing lemma is false when $\Phi$ is a rational multiple of $2\pi$. For example, if $\Phi/(2\pi) = p/q$ with $p,q \in \Ir$ and $q \neq 0$, then one can verify that $T_1^{2q}, T_2^{2q} \in Z(\vNA)$.

On $\mm_n(\vNA)$ define the trace
\begin{equation} \label{eq:taunDef}
  \tau_n=\frac{1}{n}\tr\otimes\tau,
\end{equation}
where $\tau(a):=\brAAket0a0$ is a trace on $\vNA$ which is unique if $\MF/(2\pi)$ is irrational \cite{shubinDML}, and $\tr$ is the matrix trace. That is, $\tau_n(A) = n^\inv \sum_{j=1}^n \tau(a_{jj})$.  One can check that $\tau_n$ inherits all properties of $\tau$, so $\mm_n(\vNA)$ is a ${\rm II}_1$-factor, like $\vNA$.

An important observation is that part of \cite[Proposition 2.1 (ii)]{shubinDML} carries over to the matrix valued case.  Concretely, given $A \in \mm_n(\vNA)$, we may write $A = \sum_{i,j=1}^n e_{ij} \otimes a_{ij}$, where $a_{ij} \in \vNA$. For $k,\ell \in \Ir^2$, define the matrix $A_{k\ell} \in \mm_n(\Cx)$ by
\begin{equation}
A_{k\ell}
=
\sum_{i,j=1}^n \langle k|a_{ij}|\ell\rangle \:e_{ij} .
\end{equation}
\begin{lem}\label{lem:diag}
  Let $A\in\mm_n(\vNA)$. Then, for all $k,\ell,m\in \Ir^2$, we have
  \begin{equation} \label{eq:lemDiag}
A_{k\ell} e^{-i\frac{\MF}{2} k \wedge \ell}
=
A_{(k+m),(\ell+m)} e^{-i\frac{\MF}{2}(k+m)\wedge(\ell+m)}
  \end{equation}
  where $k\wedge \ell = k_1 \ell_2 - k_2 \ell_1$ denotes the skew product in $\Ir^2$. In particular, $A_{\ell\ell} = A_{mm}$ for all $\ell,m \in \Ir^2$.
\end{lem}

\begin{proof}
This follows immediately from \cite[Proposition~2.1]{shubinDML} with $\alpha = \Phi/(2\pi)$ (compare \cite[Equation~(2.3)]{shubinDML} with \eqref{eq:uv}).
\end{proof}

\begin{remark}
  The reader should notice that Lemma~\ref{lem:diag} is a statement about a representation of $\RA$, not about $\RA$ itself. Thus, Lemma~\ref{lem:diag} yields statements that (ostensibly, at least) depend on the choice of gauge. However, we will only use the diagonal case $k=\ell$, which simply says $A_{\ell\ell}=A_{mm}$ for all $m,\ell\in \Ir^2$ and does \emph{not} depend on the gauge.
\end{remark}

We can now define the spectral distribution function (SDF) for the magnetic walk $\Wm$, which we can (and do) view as an element of $\mm_2(\vNA)$. Define a measure $dN$ on $\torus$  by
\begin{equation}
\int_\torus g \, dN
=
\tau_2(g(\Wm))
= \frac{1}{2} \sum_{s=\pm1} \braketop{0,s}{g(\Wm)}{0,s},
\end{equation}
for bounded Borel functions $g:\torus\to\Cx$. The second equality follows from the explicit form of $\tau_n$ in \eqref{eq:taunDef}.
The SDF of $\Wm$ is the accumulation function of $dN$, that is, $N(\theta) = \int_\torus \chi_{I(\theta)} \, dN$, where $I(\theta)$ denotes the arc $I(\theta) \equiv \{ \exp(i\theta') : 0 \le \theta' \le \theta\}$ and $\chi_{I(\theta)}$ is the characteristic function of $I(\theta)$. In order to talk about the point at $\theta=0$ in a coherent fashion, we adopt the convention $N(\theta + 2\pi) = N(\theta)+1$. Equivalently, we have
\begin{equation}
	N(\theta)=\tau_2\left (E_\theta\right),\quad \theta\in[0,2\pi[,
\end{equation}
where $E_\theta = \chi_{I(\theta)}(\Wm)$. It follows that the spectrum of of $\Wm$ consists of all points of growth of $N(\theta)$, i.e.
\begin{equation}
  \sigma(\Wm)=\{\exp(i\theta):N(\theta-\varepsilon)-N(\theta+\varepsilon)>0\:\forall\varepsilon>0\}.
\end{equation}

In particular, by weak continuity of the trace \cite{shubinDML}, $\exp(i\theta)$ lies in the point spectrum of $\Wm$, if and only if $N$ is discontinuous at $\theta$. Consequently, proving continuity for the spectral distribution function proves the absence of point spectrum of $\Wm$. In particular, we show that:
\begin{thm}\label{thm:SDF_cont}
  The SDF of the magnetic walk $\Wm$ defined in \eqref{eq:dirac_mag_walk} is continuous, i.e.\ $\sigma_{pp}  (\Wm) = \emptyset$.
\end{thm}

To prove this theorem we will show, that the SDF is equal to the integrated density of state (IDS), which we then show to be continuous by the Delyon--Souillard argument \cite{delyon1984remark}. There is an additional challenge present here when compared to the Hamiltonian case. Namely, the restriction of a self-adjoint operator to a subspace remains self-adjoint, but the restriction of a unitary operator need not even be normal. Thus, for our purposes, we need suitable unitary restrictions of $\Wm$ to finite $\Lambda\subset\Ir^2$, which we achieve by decoupling the walk locally around the edge of $\Lambda$ using purely off-diagonal coins. Since this does not depend on the magnetic field but only on the neighbourhood structure, we will formulate this decoupling more generally for walks of the form \eqref{eq:dirac_walk}.

In the following, let $\Lambda = \Lambda_L \subset \Ir^2$ be a square of side length $2L+1$, centered at the origin of $\Ir^2$, i.e.\
\begin{equation}\label{eq:lambda}
\Lambda=\left\{x=(x_1,x_2) \in \Ir^2 \colon (|x_1|\leq L)\, \land\, (|x_2|\leq L)\right\},
\end{equation}
for some $L>0$.  The boundary $\partial\Lambda$ of $\Lambda$ is then given by the set
\begin{equation}
\partial\Lambda=\left\{x\in\Lambda\colon (|x_1|= L)\, \lor\, (|x_2|= L)\right\}.
\end{equation}
For the unitary decoupling, we need a slightly different bounding set, namely the union of the upper and the right edges of $\Lambda_L$ and the lower and the left edges of $\Lambda_{L+1}$. To that end, put
\begin{equation}
\Delta\Lambda
\equiv\partial\left\{x=(x_1,x_2) \in \Ir^2 \colon \left(-L-1\leq x_1\leq L\right)\,\land\,(-L-1\leq x_2\leq L)\right\}
\end{equation}
Denote by $P_L : \ell^2(\Ir)^2 \otimes \Cx^2 \to \ell^2(\Lambda_L) \otimes \Cx^2$ the canonical projection.

\begin{lem}
	Let $L \in \Nl$, let $\Lambda = \Lambda_L$ be as in \eqref{eq:lambda} and let $W$ be a walk of the form \eqref{eq:dirac_walk}. Then there is a walk $W_d$, such that $W_L=P_L W_d P_L^*$ is unitary on $P_L\HH$ and $W-W_d$ differs from zero only on the set
	\begin{equation}
	\left(\Delta\Lambda\right)_2=\left\{x\in\Ir^2 : \dist(x,\Delta\Lambda)\leq 2 \right\}.
	\end{equation}
\end{lem}

\begin{proof}
	Define the coin operator $C_d$, which acts locally as
	\begin{equation}
	C_d(x)=\begin{cases} \sigma_1    &   x\in\Delta\Lambda    \\  \idty_2   & \text{else}\end{cases}
	\end{equation}
	with $\sigma_1$ being the first Pauli matrix. Moreover replace $S_\alpha,\alpha=1,2$ by
	\begin{equation}
	S_\alpha\mapsto\widetilde S_\alpha= S_\alpha^\uparrow C_dS_\alpha^\downarrow,
	\end{equation}
	where $S_\alpha^{\uparrow/\downarrow}$ shift the spin-up component in the positive and the spin-down component in the negative direction, respectively.
	The resulting walk
	\begin{equation}
	W_d=\left(S_1^\uparrow C_dS_1^\downarrow\right) C_{1}\left(S_2^\uparrow C_dS_2^\downarrow\right) C_{2}
	\end{equation}
	is decoupled between every connection of $\Lambda$ and $\Lambda^c$, i.e. it fulfills the requirements in the Lemma, which is straightforward to check.
\end{proof}
\begin{remark}
  It follows from this proof that the asymmetry in the definition of $\Delta\Lambda$ is needed because of the left/right (up/down) asymmetry of $S_\alpha^\uparrow C_dS_\alpha^\downarrow$. Instead, one could have insisted on a symmetric $\Delta\Lambda$ and had to choose $\widetilde S_\alpha$ asymmetric.
\end{remark}

Denoting by $\abs\Lambda$ the number of points in a finite set $\Lambda$ and by $(\partial\Lambda)_n$
the elements of $\Lambda$ with distance less than $n\geq1$ from the boundary, we write $\Lambda\to\infty$ whenever for some fixed $n\geq1$
\begin{equation}
  \frac{\abs{(\partial\Lambda)_n}}{\abs\Lambda}\to0.
\end{equation}
In particular, we note that $\Lambda_L \to \infty$ as $L \to \infty$. This together with the decoupling above allows us to define the density of states measure (DOS) for the magnetic walk $\Wm$ as the accumulation function of the vague limit of eigenvalue counting measures of the unitary truncations. That is, given $L$, define a measure $dk_L$ by $\int_\torus g \, dk_L = (2|\Lambda_L|)^{-1}\tr(g(W_L))$. Then, the DOS, denoted $dk$, is given by
\begin{equation}
\int_{\torus} g \, dk
=
\lim_{L \to \infty} \int_\torus g \, dk_{L}
=
\lim_{L \to \infty} \frac{1}{2|\Lambda_L|} \tr(g(W_{L}))
\end{equation}
for bounded Borel $g$,
and $k(\theta)$ is given by
\begin{equation}
  k(\theta)
  = \int_\torus \chi_{I(\theta)} \, dk
=   \lim_{L\to\infty}\frac{1}{2|\Lambda_L|} \tr\left( \chi_{I(\theta)}(W_{L}) \right).
\end{equation}
The existence of the limit in the definition of $k(\theta)$ follows from the arguments in \cite{pastur1980}. Important for our purpose is that $\theta\mapsto k(\theta)$ is continuous which follows from a proof similar to that in \cite{delyon1984remark}:
\begin{prop}\label{prop:IDS_cont}
  Assume that the coins $C_j$ are either not completely diagonal or not completely off-diagonal. Then $\theta\mapsto k(\theta)$ is continuous.
\end{prop}

\begin{proof}
  To prove the continuity of $\theta\mapsto k(\theta)$ it is enough to prove that
  \begin{equation}
    |\Lambda_L|^\inv\tr(\chi_{\{\exp(i\theta)\}}(W_L))\to0,\quad\theta\in\torus,
  \end{equation}
  as $L \to \infty$.  Let $\Lambda=\Lambda_L$ be defined as in \eqref{eq:lambda} and without loss, assume that the coins $C_j$ are not completely diagonal. Then, inside $\Lambda$ a solution $\psi$ of the generalized eigenvalue equation $W\psi=z\psi$ is uniquely determined by its values on the set
  \begin{equation}\label{eq:ids_cont_boundary}
    \{x\in\Lambda:x_1=-L\}\cup\{x\in\Lambda:x_2=-L,L-1,L\}.
  \end{equation}
  This can be seen by explicitly evaluating the action of a walk of the form \eqref{eq:dirac_walk} with coins $C_1$ and $C_2$ at position $x\in\Ir^2$:
  \begin{align}
\label{eq:WCoordActionUp}
 (W\psi)(x,+)&=c^1_{11}(x-e_1)\left(c^2_{11}(x-e_1-e_2)\psi(x-e_1-e_2,+)+c^2_{12}(x-e_1-e_2)\psi(x-e_1-e_2,-)\right)\\
 \nonumber
 &\quad+c^1_{12}(x-e_1)\left(c^2_{21}(x-e_1+e_2)\psi(x-e_1+e_2,+)+c^2_{22}(x-e_1+e_2)\psi(x-e_1+e_2,-)\right) \\
 \label{eq:WCoordActionDown}
    (W\psi)(x,-)&=c^1_{21}(x+e_1)\left(c^2_{11}(x+e_1-e_2)\psi(x+e_1-e_2,+)+c^2_{12}(x+e_1-e_2)\psi(x+e_1-e_2,-)\right)\\
    \nonumber
    &\quad+c^1_{22}(x+e_1)\left(c^2_{21}(x+e_1+e_2)\psi(x+e_1+e_2,+)+c^2_{22}(x+e_1+e_2)\psi(x+e_1+e_2,-)\right).
  \end{align}

The algorithm to determine a solution $\psi$ for $W\psi=z\psi$ inside $\Lambda$ from its values on the set \eqref{eq:ids_cont_boundary} is as follows: first, by \eqref{eq:WCoordActionUp}, the values of a solution $\psi$ on the set \eqref{eq:ids_cont_boundary} immediately determine the values of $\psi(x,+)$ for $\{x:x_1=-L+1,x_2=-L+2,\dots,L-2\}$, see the left picture in Figure \ref{fig:IDS_cont}.

In the second step the values of $\psi(x,-)$ are determined subsequently for $x_1=-L+1$ ,$x_2=L-2,L-3\dots,-L+1$. The key point is this: if $\psi$ solves the generalized eigenvalue equation, then, since neither $C_1$ nor $C_2$ is completely diagonal, $\psi(x,-)$ is uniquely determined by $\psi(x,+)$, $\psi(x-e_{1}+e_{2},-)$, and $\psi(x+2e_{2},\pm)$ via \eqref{eq:WCoordActionDown}. Concretely,
\begin{equation} \label{eq:downSpinRecursion}
\begin{split}
\psi(x,-)
& =
\frac{1}{c_{12}^2(x)}\bigg[ \frac{1}{c_{21}^1(x+e_{2})} \Big[ z\psi(x-e_{1}+e_{2},-)
-
c_{22}^1(x+e_{2}) [c_{21}^2(x+2e_{2})\psi(x+2e_{2},+)  \\
& \qquad\qquad\qquad + c_{22}^2(x+2e_{2})\psi(x+2e_{2},-)]
 - c_{11}^2(x)\psi(x,+) \Big] \bigg].
\end{split}
\end{equation}

\begin{figure}[t]
  \begin{center}
    \includegraphics[width=.7\textwidth]{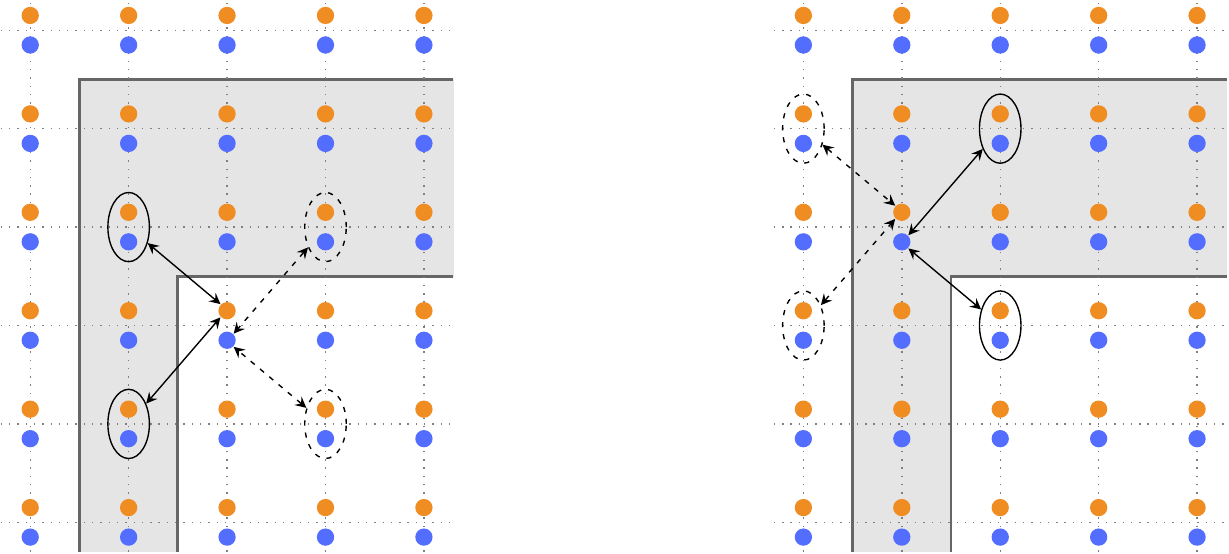}
    \caption{\label{fig:IDS_cont}The algorithm to determine a generalized eigenfunction $\psi$ of $\Wz$ of the form \eqref{eq:dirac_walk} in the proof of Proposition \ref{prop:IDS_cont}. At each lattice site $x\in\Ir^2$ the internal degrees of freedom of $\psi$ are distinguished by the colouring orange for $\psi(x,+)$ and blue for $\psi(x,-)$. The arrows indicate which degrees of freedom are coupled by the action of $W$ (compare with \eqref{eq:WCoordActionUp} and \eqref{eq:WCoordActionDown}), and the shaded area indicates the boundary set \eqref{eq:ids_cont_boundary}. Note that the diagonal arrows to the left and the right of $x$ disappear for diagonal coins at $x-e_1$ and $x+e_1$, respectively, whereas the anti-diagonal arrows disappear for off-diagonal coins.
    Left: in the first step of the algorithm $\psi(x,+)$, $x_1=-L+1,\abs{x_2}\leq L-2$ is obtained from the boundary values $\psi(x)$, $x_1=-L,\abs{x_2}\leq L-1$. Right: in the second step, $\psi(x,-)$, $x_1=-L+1$ is determined successively for $x_2=L-2,L-3,\dots,-L+2$ by solving the eigenvalue equation $W\psi=z\psi$ on the boundary $x_1=-L,x_2=L-1,L-2,\dots,-L+1$. }
  \end{center}
\end{figure}

But the value of $\psi(x,+)$ on the strip $x_1=-L+1$, $x_2=L-2,L-3,\dots,-L+1$ we just calculated, and they determine $\psi(x,-)$ with $x_1=-L+1$ subsequently for $x_2=L-2,L-3,\dots,-L+1$, see the right picture in Figure \ref{fig:IDS_cont}.

Repeating this algorithm on each strip from the left to the right determines $\psi$ completely inside $\Lambda$ from its values on the set \eqref{eq:ids_cont_boundary}. We can thus repeat the argument given in \cite{delyon1984remark}, namely
\begin{equation}
  \dim(\ker(W_L - e^{i\theta}))
  \leq2\abs{\eqref{eq:ids_cont_boundary}} =  2(8L+1)  < 8(2L+1)
\end{equation}
where the prefactor takes into account that we have to take one strip on the left (or right) side and two on the top/botton, and that the internal degree of freedom is two-dimensional. Then, since $\abs\Lambda=\left(2L+1\right)^2$ we have
\begin{equation}
  \frac1{\abs\Lambda}\tr(\chi_{\{\exp(i\theta)\}}(W_L))<\frac{8}{2L+1}\to0
\end{equation}
as $\Lambda\to\infty$.
\end{proof}
\begin{remark}
  Note that under the assumption that the coins are not completely \emph{off-diagonal} the proof is completely analogous after replacing the set \eqref{eq:ids_cont_boundary} by
  \begin{equation}
    \{x\in\Lambda:x_1=-L\}\cup\{x\in\Lambda:x_2=-L,-L+1,L\}.
  \end{equation}
  In this case we solve for $\psi(x,-)$ as in \eqref{eq:downSpinRecursion} but from the bottom of $\Lambda$ instead of from the top, i.e. instead of \eqref{eq:downSpinRecursion} we solve \eqref{eq:WCoordActionDown} as
  \begin{equation}
    \begin{split}
    \psi(x,-)
        & =
        \frac{1}{c_{22}^2(x)}\bigg[ \frac{1}{c_{22}^1(x-e_{2})} \Big[ z\psi(x-e_{1}-e_{2},-)
        -
        c_{21}^1(x-e_{2}) [c_{11}^2(x-2e_{2})\psi(x-2e_{2},+)  \\
        & \qquad\qquad\qquad + c_{12}^2(x-2e_{2})\psi(x-2e_{2},-)]
        - c_{21}^2(x)\psi(x,+) \Big] \bigg],
    \end{split}
  \end{equation}
  and apply the algorithm ``upwards'' instead of ``downwards''. Both requirements in Proposition \ref{prop:IDS_cont} are clearly met in the magnetic walk model where the coins are given in \eqref{eq:symGaugeQPcoins}.
\end{remark}

Coming back to the magnetic walk $W=\Wm$ which in the gauge \eqref{eq:symmetricGauge} indeed is of the form \eqref{eq:dirac_walk}, see Remark \ref{rem:sym_gauge_coins}, Theorem \ref{thm:SDF_cont} follows from
\begin{prop}
  For all $\theta\in[0,2\pi[$ we have
  \begin{equation}
    k(\theta)=N(\theta).
  \end{equation}
\end{prop}

For the proof of the equivalence of the IDS and the SDF we need the following variant of the celebrated Levy continuity theorem which is a special case of \cite[Theorem 4.2.5]{applebaum} for the compact group $\torus$:
\begin{lem}\label{lem:levy}
  Let $\mu_n$, $n\in\Nl$ be a sequence of probability measures on $\torus$. Then $\mu_n\xrightarrow w\mu$ weakly     if and only if $\widehat \mu_n\to\widehat\mu$ pointwise, where $\widehat \mu_n$ denotes the Fourier transform of $\mu_n$, which is defined by $\widehat \mu_n(t) =\int_\torus e^{it\theta}\mu_n(d\theta)$, $t \in \Ir$.
\end{lem}

\begin{proof}[Proof of Theorem \ref{thm:SDF_cont}]
 Given $t\in\Ir$ and $L \in \Nl$, we have
  \begin{equation}
    \braketop{x,s}{W_{L}^t }{x,s}=\braketop{x,s}{W^t}{x,s},
  \end{equation}
  for $s=\pm1$ and $x\in\Lambda_L$ with $\dist(x,\partial \Lambda_L ) > 2|t|$. This implies that
  \begin{equation}
    \lim_{L\to\infty}\frac1{2|\Lambda_L|}\sum_{x\in\Lambda}\sum_s\left(\braketop{x,s}{W_{L}^t}{x,s}-\braketop{x,s}{W^t}{x,s}\right)=0.
  \end{equation}
By Lemma \ref{lem:diag}, we have $\braketop{x,s}{W^t}{x,s}= \braketop{0,s}{W^t}{0,s}$ for all $x \in \Ir^2$ and $s = \pm1$, and hence
  \begin{equation}\label{eq:equaltraces}
    \lim_{L\to\infty}\frac1{2|\Lambda_L|}\tr\left(W_{L}^t\right)=\tau_2(W^t).
  \end{equation}
  The expression $(2|\Lambda_L|)^{-1}\tr\left(W_{L}^t\right)$ on the left side of this expression is equal to the Fourier transform of the measure $dk_{L}$ whereas the right side is the Fourier transform of the measure $dN$. The statement then follows from Lemma \ref{lem:levy}.
\end{proof} 
\section*{Acknowledgements}
C. Cedzich acknowledges support by the projet PIA-GDN/QuantEx P163746-484124 and by {\em DGE -- Minist\`{e}re de l'Industrie}.

J. Fillman thanks Svetlana Jitomirskaya for helpful conversations.

T. Geib acknowledge support from the DFG SFB 1227 DQmat.

A. H. Werner thanks the VILLUM FONDEN for its support with a Villum Young Investigator Grant (Grant No. 25452) and its support via the QMATH Centre of Excellence (Grant No. 10059).

\bibliography{mwalksbib}

\end{document}